\documentclass[runningheads]{llncs}
\usepackage[T1]{fontenc}

\usepackage[symbol]{footmisc}
\usepackage{filecontents}
\usepackage{graphicx}
\usepackage{hhline}
\usepackage{xspace}
\usepackage{balance}
\usepackage{ulem}
\usepackage{url}

\usepackage{colortbl}
\usepackage{multirow}
\usepackage{multicol}
\usepackage{capt-of}

\usepackage{tikz}

\usepackage{amsmath}
\usepackage{enumerate}
\usepackage{enumitem}
\usepackage{orcidlink}

\usepackage{algorithmicx}
\usepackage{algorithm}
\usepackage{algpseudocode}
\usepackage{varwidth}

\usepackage{amsfonts}
\usepackage{textcomp}
\usepackage{gensymb}
\usepackage{makecell}
\usepackage{tabularx, booktabs}
\usepackage{array}
\usepackage{times}
\usepackage{fullpage}
\usepackage{dsfont}
\usepackage{epsfig}
\usepackage{microtype}
\usepackage{mathtools}
\usepackage{cleveref}
\usepackage[disable]{todonotes}
\usepackage{bm}

\newcommand{\allprocesses}{N}
\newcommand{\system}{\textsc{BBCA-Chain}\xspace}

\newcommand{\primitive}{BBCA}
\newcommand{\complete}{complete\xspace}
\newcommand{\completed}{completed\xspace}
\newcommand{\completes}{completes\xspace}
\newcommand{\Complete}{Complete\xspace}
\newcommand{\completeCert}{completeCert\xspace}
\newcommand{\completion}{completion\xspace}
\newcommand{\prim}{\primitive\xspace}
\newcommand{\bbca}{\textsc{BBCA}\xspace}
\newcommand{\bbcaadopt}[3]{\langle\textsc{BBCA-adopt}, #1, #2, #3\rangle}
\newcommand{\bbcanoadopt}[1]{\langle\textsc{BBCA-noadopt}, #1\rangle}
\newcommand{\bbcacomplete}[3]{\textsc{BBCA-complete}(#1, #2, #3)}
\newcommand{\newview}{new-view\xspace}
\newcommand{\leader}{backbone\xspace}
\newcommand{\regular}{data\xspace}
\newcommand{\cbc}{BCB\xspace}
\newcommand{\rbc}{BRB\xspace}

\newcommand{\newinbbca}[1]{{\color{blue}#1}}
\newcommand{\removedinbbca}[1]{{\color{red}#1}}

\newcommand{\pendingMsgs}{\mathit{pendingMsgs}}
\newcommand{\bbcaState}{\mathit{bbcaState}}
\newcommand{\view}{\mathit{view}}
\newcommand{\viewTimer}{\mathit{viewTimer}}
\newcommand{\committed}{\mathit{completed}}
\newcommand{\adopted}{\mathit{adopted}}
\newcommand{\skipped}{\mathit{noadopted}}
\newcommand{\lastCommitted}{\mathit{lastCommitted}}
\newcommand{\blockref}{\mathit{ref}}
\newcommand{\prev}{\mathit{prevView}}
\newcommand{\backbone}{\mathit{backbone}}
\newcommand{\newviewblocks}{\mathit{newViewBlocks}}
\newcommand{\newViewBlock}{\mathit{newViewBlock}}
\newcommand{\finalized}{\mathit{finalized}}

\newcommand{\cert}{\mathit{cert}}
\newcommand{\sig}{\mathit{sig}}

\newcommand{\receivedEcho}{\mathit{receivedEcho}}
\newcommand{\receivedReady}{\mathit{receivedReady}}

\newcommand{\echo}{\mathit{echo}}
\newcommand{\ready}{\mathit{ready}}
\newcommand{\abort}{\mathit{abort}}
\newcommand{\nEcho}{\mathit{nEcho}}
\newcommand{\nReady}{\mathit{nReady}}

\newcommand{\bid}{\mathit{bid}}
\newcommand{\sender}{\mathit{sender}}
\newcommand{\broadcast}[1]{\textsc{broadcast}(#1)}
\newcommand{\initMsg}[2]{\langle\textsc{Init}, #1, #2 \rangle}
\newcommand{\echoMsg}[2]{\langle\textsc{Echo}, #1, #2 \rangle}
\newcommand{\readyMsg}[2]{\langle\textsc{Ready}, #1, #2 \rangle}

\newcommand{\bstate}{\mathit{bstate}}
\newcommand{\vstate}{\mathit{mstate}}
\newcommand{\vstates}{\mathit{mstates}}
\newcommand{\InlineComment}[1]{\State{\color{gray}// #1}}

\newcommand{\letvar}[2]{\textbf{let}\,\,{#1} = {#2}}
\newcommand{\ifthen}[2]{\textbf{if} #1 \textbf{then} #2}
\algrenewcommand\algorithmicindent{1.0em}

\newcommand{\true}{\textsc{true}\xspace}
\newcommand{\false}{\textsc{false}\xspace}

\newcommand{\event}[3]{
    \ifthenelse
    {\equal{#3}{}}
    {\ap{#1.\textrm{#2}}}
    {\ap{#1.\textrm{#2} \mid #3}}
}

\algnewcommand\Instance[2]{\State #1, \textbf{instance} #2}
\algnewcommand\InstanceSystem[3]{\State #1, \textbf{instance} #2, \textbf{system} #3}
\algnewcommand\Trigger[3]{\State \textbf{trigger} $\event{#1}{#2}{#3}$}
\algnewcommand\Schedule[1]{\State \textbf{schedule} $#1$}
\algnewcommand\Cancel[1]{\State \textbf{cancel} $#1$}
\algnewcommand\Broadcast[1]{\State \textbf{broadcast} $#1$}
\algnewcommand\Import[1]{\State \textbf{import} $#1$}

\algnewcommand\Not{\textbf{ not }}
\algnewcommand\AndT{\textbf{ and }}
\algnewcommand\OrT{\textbf{ or }}
\algnewcommand\In{\textbf{ in }}

\algsetblockdefx[Implements]{Implements}{EndImplements}{}{}{\textbf{Implements:}}{}
\algsetblockdefx[Uses]{Uses}{EndUses}{}{}{\textbf{Uses:}}{}
\algsetblockdefx[Init]{Init}{EndInit}{}{}{\textbf{Init:}}{}
\algsetblockdefx[Parameters]{Parameters}{EndParameters}{}{}{\textbf{Parameters:}}{}
\algsetblockdefx[Struct]{Struct}{EndStruct}{}{}[1]{\textbf{Struct} $#1$ \textbf{contains}}{}

\algsetblockdefx[Upon]{Upon}{EndUpon}{}{}[1]{\textbf{upon} $#1$ \textbf{do}}{}
\algsetblockdefx[UponExists]{UponExists}{EndUponExists}{}{}[2]{\textbf{upon exists} $#1$ \textbf{such that} $#2$ \textbf{do}}{}
\algsetblockdefx[UponCondition]{UponCondition}{EndUponCondition}{}{}[1]{\textbf{upon} $#1$ \textbf{do}}{}
\algsetblockdefx[UponReceiving]{UponReceiving}{EndUponReceiving}{}{}[1]{\textbf{upon receiving} $#1$ \textbf{do}}{}
\algsetblockdefx[UponEvent]{UponEvent}{EndUponEvent}{}{}[1]{\textbf{upon event} $#1$ \textbf{do}}{}
\algsetblockdefx[UponMsg]{UponMsg}{EndUponMsg}{}{}[2]{\textbf{upon receiving} $#1$ \textbf{from} $#2$ \textbf{do}}{}
\algsetblockdefx[UponReceivingX]{UponReceivingX}{EndUponReceivingX}{}{}[3]{\textbf{upon receiving #1} $#2$ \textbf{from} $#3$ \textbf{do}}{}

\algsetblockdefx[Process]{Process}{EndProcess}{}{}[2]{\textbf{process} \emph{#1}($#2$) \textbf{:}}{}
\algsetblockdefx[Function]{Function}{EndFunction}{}{}[2]{\textbf{function} \emph{#1}($#2$) \textbf{:}}{}
\algsetblockdefx[Procedure]{Procedure}{EndProcedure}{}{}[2]{\textbf{procedure} \emph{#1}($#2$) \textbf{is}}{}
\algsetblockdefx[FunctionNoArgs]{FunctionNoArgs}{EndFunctionNoArgs}{}{}[1]{\textbf{function} \emph{#1}() \textbf{:}}{}
\algsetblockdefx[Interface]{Interface}{EndInterface}{}{}[2]{\textbf{interface} \emph{#1}($#2$) \textbf{:}}{}

\algdef{SE}[DoUntil]{DoUntil}{EndDoUntil}{\textbf{do}}[1]{\textbf{until} $#1$}
\algsetblockdefx[Struct]{Struct}{EndStruct}{}{}[1]{\textbf{Struct} $#1$ \textbf{contains}}{}
\algsetblockdefx[IfExists]{IfExists}{EndIfExists}{}{}[2]{\textbf{if exists} $#1$ \textbf{such that} $#2$ \textbf{then}}{\textbf{end if}}
\algsetblockdefx[While]{While}{EndWhile}{}{}[1]{\textbf{while} $#1$ \textbf{do}}{\textbf{end while}}
\algsetblockdefx[NTimes]{NTimes}{EndNTimes}{}{}[1]{\textbf{for} $#1$ \textbf{times do}}{\textbf{end for}}
\algsetblockdefx[For]{For}{EndFor}{}{}[3]{\textbf{for} $#1 \in #2..#3$ \textbf{do}}{\textbf{end for}}
\algsetblockdefx[ForAll]{ForAll}{EndForAll}{}{}[2]{\textbf{for} $#1 \in #2$ \textbf{do}}{\textbf{end for}}

\graphicspath{{figures/}}

\begin{document}
\title{\system: Low Latency, High Throughput BFT Consensus on a DAG}
\titlerunning{Low Latency, High Throughput BFT Consensus on a DAG}

\author{Dahlia Malkhi \orcidlink{0000-0002-7038-7250} \and
Chrysoula Stathakopoulou\orcidlink{0000-0002-3958-9735} \and
Maofan Yin\orcidlink{0009-0006-2717-8150}}
\authorrunning{D. Malkhi et al.}
\institute{Chainlink Labs}
\maketitle
\begin{abstract}
This paper presents a partially synchronous BFT consensus protocol powered by BBCA, a lightly modified Byzantine Consistent Broadcast (\cbc) primitive.
BBCA provides a Complete-Adopt semantic through an added probing interface to allow either aborting the broadcast by correct nodes or exclusively, adopting the message consistently in case of a potential delivery.
It does not introduce any extra types of messages or additional communication costs to \cbc.

BBCA is harnessed into BBCA-CHAIN to make direct commits on a chained backbone of a causally ordered graph of blocks, without any additional voting blocks or artificial layering.  With the help of Complete-Adopt, the additional knowledge gained from the underlying \cbc completely removes the voting latency in popular DAG-based protocols.  At the same time, causal ordering allows nodes to propose blocks in parallel and achieve high throughput.  BBCA-CHAIN thus closes up the gap between protocols built by consistent broadcasts (e.g., Bullshark) to those without such an abstraction (e.g., PBFT/HotStuff), emphasizing their shared fundamental principles.

Using a Bracha-style \cbc as an example, we fully specify BBCA-CHAIN with simplicity, serving as a solid basis for high-performance replication systems (and blockchains).
\end{abstract}

\section{Introduction}
A consensus protocol allows a network of nodes to consistently commit a sequence of values (referred to as \emph{blocks}).
Usually, practical solutions assume a \emph{partially synchronous} network
that may suffer from temporary outages or periods of congestion, where
protocols can tolerate Byzantine faulty nodes exhibiting arbitrary behavior
as long as a \emph{quorum} of two-thirds is correct.

After decades, this line of research has reached the known communication complexity lower bound that is linear~\cite{HotStuff-19,malkhi2023hotstuff} for steady operations and upon reconfiguration.
However, these solutions operate in a sequential manner, extending the committed prefix one proposed block at a time.
Hence, throughput is limited by the (underutilized) proposer’s computational power, as it must wait for a round-trip to collect a quorum before moving on.

To fully utilize the idle time in which a proposer waits for the next sequential progress,
recent works~\cite{GagolLSS19,DAG-rider,Narwhal,spiegelman2022bullshark,fino} harness Lamport's causal communication to allow blocks to be injected to the network in parallel to finalize their commit ordering.
Each block references one or more previous blocks and the protocol relies on causal ordering,
so that when nodes deliver the block, they have already delivered
the blocks in its entire ancestry on a Direct Acyclic Graph (DAG).
Then, the consensus logic ``rides'' on the DAG, so that nodes interpret it locally for a sequentially committed \emph{backbone} (the chain made from blocks with crowns in Figure~\ref{fig:backbone}). Each backbone block also commits other non-backbone blocks it transitively references. Systems based on such a design have demonstrated significant throughput gains.

\begin{figure}[th]
  \centering
  \includegraphics[width=.45\textwidth]{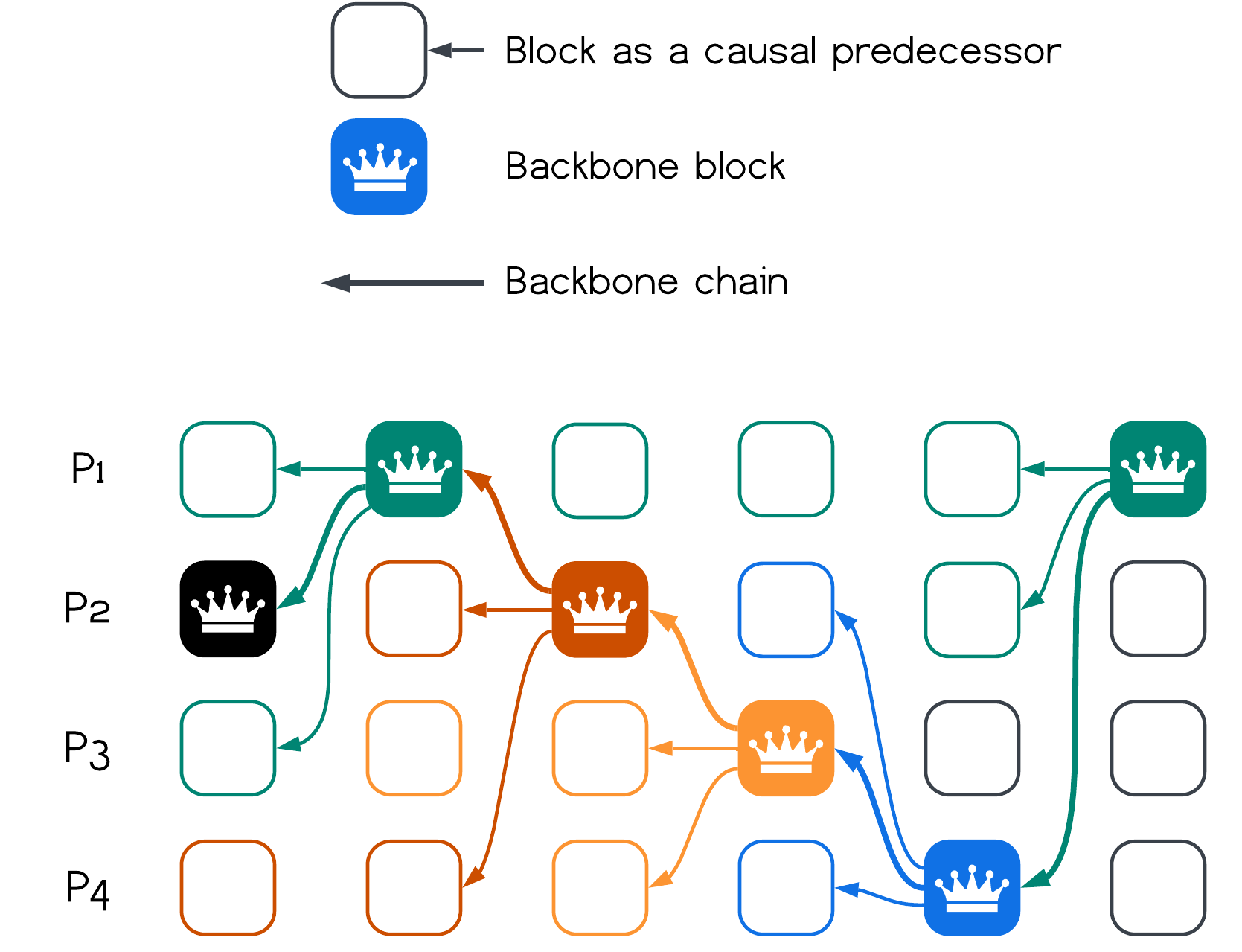}
  \caption{A sequenced backbone on a DAG (thick arrows and solid blocks) and referenced blocks that commit with it (thin arrows and hollow blocks).}
  \label{fig:backbone}
\end{figure}

However, existing DAG-riding solutions come with a relatively complicated logic and suffer from several times higher latency.
For example, a solution like Bullshark~\cite{spiegelman2022bullshark} takes a minimum of 4 chained blocks to commit one block,
as illustrated on the left of Figure~\ref{fig:tricks}:
a regular block at the head of the chain needs to first be delivered into the DAG; secondly it must be followed and referenced by a leader block which is restricted to odd layers of the DAG, hence may take two more blocks;
third,
the proposal must be referenced by a quorum of vote blocks that need to be on a subsequent, even layer dedicated to voting.
Each block in this procedure has to be delivered by a Byzantine Consistent Broadcast (\cbc) instance to guarantee a consistent view at all nodes.
Internally, \cbc already consists of multiple rounds of messages,

This paper introduces \system, a new DAG-based consensus solution that enables the same high throughput with much lower latency.
\system{} substantially simplifies the overall consensus logic: it commits a backbone block via a single \bbca broadcast,
a variant of \cbc,
without requiring vote blocks or waiting to broadcast into the artificial layers.
Apart from using fewer steps than existing protocols, it also places fewer constraints on how the DAG could be grown.

\begin{figure*}[th]
  \centering
  \includegraphics[width=\textwidth]{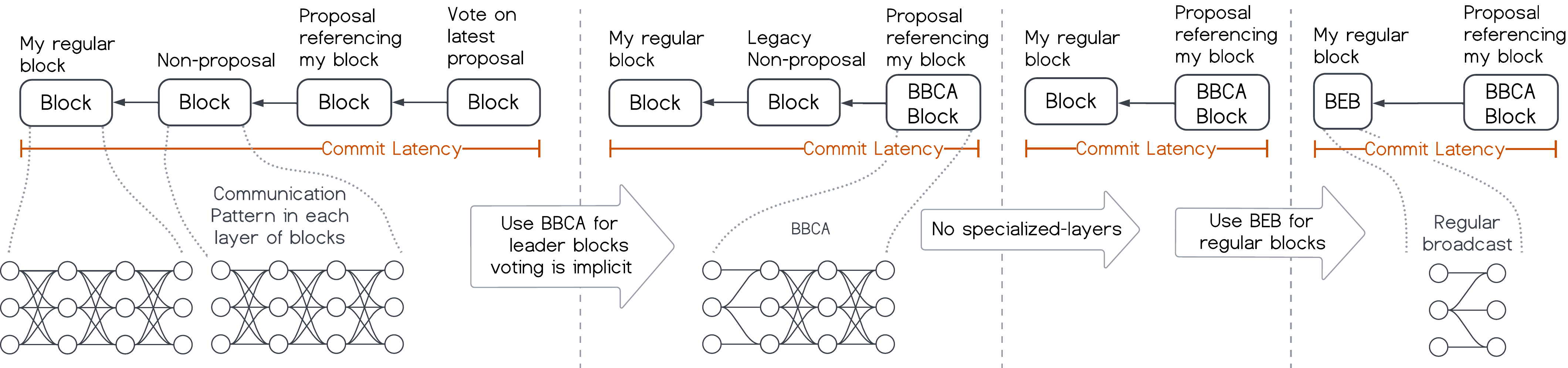}
  \caption{Reducing Bullshark commit latency (left) via three mechanisms.}
  \label{fig:tricks}
\end{figure*}

\system achieves these improvements via three mechanisms depicted in Figure~\ref{fig:tricks} from left to right:

Instead of encoding the voting process through the DAG layers, it removes vote blocks altogether because \bbca addresses voting internally.
More specifically, by peeking inside \cbc, we observe that it is already a multi-phase process similar to
what other consensus protocols (such as PBFT~\cite{PBFT}) do in their main logic.
We add a shim on top of \cbc to create \bbca (Byzantine Broadcast with \Complete-Adopt),
a new primitive that supports \emph{probing} the local state of the broadcast protocol.
The key new feature is a \Complete-Adopt interface: if a broadcast is \completed by any non-faulty node running the \bbca protocol, then a quorum probing via
\bbca-probe will \emph{adopt} it. Said differently, if a quorum probes and gets a \textsc{\bbca-noadopt} response,
then it is guaranteed that no \completion of the \bbca broadcast will ever occur to any (non-faulty) node.
Through the \Complete-Adopt interface from BBCA, one can extract this useful invariant to directly drive the committed backbone.
As probing is done locally with a simple check of some protocol state variables, it does not incur any additional communication.

As no vote blocks are needed, \system does not need special layering, which further reduces the latency.
As depicted in the middle of Figure~\ref{fig:tricks}, a block can be immediately referenced by a leader block and become committed when the leader block is \completed by \prim-broadcast.
Finally, only leader blocks need to use \bbca. All other blocks can be broadcast with a ``Best-Effort Broadcast'' that simply sends the same message to all nodes via the inter-node channels (denoted as ``BEB'' block on the right of Figure~\ref{fig:tricks}).

The table below summarizes the latency in terms of consecutive network-trips of \system compared with state-of-art DAG solutions like Bullshark. The implementation of the \bbca broadcast is assumed to use an all-to-all protocol, since there is only a single \bbca invocation per view (see more details on broadcasts in the body of the paper). 

\begin{table}[h!]
\small
\centering
\begin{tabular}[width=\textwidth]{r|c|c}
& \system & Bullshark \\ \hline\hline
Broadcast primitives & \makecell{BEB (1 trip) /\\ \bbca (3 trips)} & \cbc (3 trips)\\ \hline
leader block & $1 \times$ \bbca & $2 \times$ \cbc \\ \hline
non-leader block & $1 \times$ BEB $\ +\ 1 \times$ \bbca & $4 \times$ \cbc \\
\end{tabular}
\vspace{0.2cm}
\caption{Latency to commit a block on a DAG.}
    \label{tab:consensus-latency}
\end{table}

\paragraph{The Benefits of Causal Ordering.}
\system is an interesting mid-point between non-DAG and DAG-based consensus solutions,
that enjoys the best of both worlds: the same latency as in those non-DAG protocols, the high throughput from
a better utilization of processing power while waiting for backbone's progress, and the simplicity of checking a DAG
to know the commitment progress.

In steady state,
the way the sequenced backbone is formed in \system resembles traditional non-DAG consensus protocols.
The difference is that in the latter,
the network is underutilized as the protocol only disseminates blocks on the linear backbone, whose
rate of progress is bounded by a quorum round-trip time, whereas DAG-based protocols can keep injecting
non-backbone blocks, which carry useful payload, without waiting.
Furthermore, the leader has to broadcast a (potentially large) block combining all the requests.
Previous studies~\cite{mirbft} have demonstrated that a single leader becomes a throughput bottleneck, even if it broadcasts blocks which only include references to transactions.
\system leverages causal communication to allow nodes to inject blocks in parallel.
Then the leader sends a (smaller) block that simply references other blocks.
In addition, \system benefits from being able to parallelize the validation of blocks, offloading the
computational burden from the backbone.

At the same time, the way \system handles failures borrows the simplified logic of DAG-based consensus.
When leaders fail, the causal history allows nodes to accept leader proposals without reasoning about (un)locking.
Compared to existing DAG-based protocols,
\system further simplifies the solution by requiring nodes only to inject \newview blocks into the DAG, foregoing vote blocks.

Foregoing block-votes on the DAG has two additional far-reaching benefits.
\begin{enumerate}
\item It avoids the extra latency introduced in previous DAG-based solutions.
Rather, \system makes use of acknowledgment messages which are in any case sent inside the \bbca broadcast protocol as votes.

\item Whereas votes must be non-equivocating, blocks do not need to.
Because we got rid of vote blocks, there is no need to use \cbc for non-leader blocks.
Instead, regular blocks can be uniquely ordered via the backbone utilizing only causality.
This change allows further reduction in latency.
\end{enumerate}

\section{Technical Overview}
\label{sec:overview}

\paragraph{Consensus.}

In the celebrated State Machine Replication (SMR) approach, a network of validator nodes form consensus on a growing totally ordered log (or ``chain'') of transactions.
The focus in this paper is on practical settings in which during periods of stability, transmission delays on the network are upper bounded, but occasionally the system may suffer unstable periods. This is known as the \emph{partial synchrony} model~\cite{DworkLS88}.
The strongest resilience against Byzantine faults which is possible in this setting is $f$-tolerance, where $f < n/3$ for a network of $n$ validators.


\paragraph{Parallel Leader Proposals.}
Concensus protocols require a \emph{leader} node to propose the next transaction to be added to the chain
and broadcast it to all validator nodes.
Then, validators run a voting protocol to \emph{commit} the leader proposal.
In practice, to amortize the computation and communication per voting decision,
leaders bundle multiple transactions in blocks in their proposals.
However, this mechanism alone is not enough to ensure high throughput, as the larger the set of validators,
the longer it takes for the leader to broadcast the proposal.
During this time, the rest of nodes remain idle.
To enable throughput scalability, we borrow from recent consensus protocols (e.g.~\cite{smrmadesimple,RCC,Narwhal,spiegelman2022bullshark}) and allow parallel leader proposals.

\paragraph{Causal Broadcast.}
A Direct Acyclic Graph (DAG) captures a transport substrate that guarantees causal ordering of block delivery.
Whenever a node is ready to broadcast a block, the block contains references to some locally delivered blocks.
The consensus logic operates above the transport level,
such that blocks are handled (received and generated) in causal order.

\paragraph{The Backbone.}
The general approach we take for DAG-based consensus is forming a totally ordered \emph{backbone} of leader blocks
which we therefore call \leader blocks  (Figure~\ref{fig:backbone}).
When a \leader block becomes committed, every leader block in its causal ancestry becomes committed as well.
By walking from the earliest uncommitted leader block in the backbone towards the latest committed leader block,
one can also commit into a total order those blocks outside the backbone but causally referenced by the \leader blocks, according
to any predefined, deterministic DAG traversal order.
Forming consensus on the backbone itself ``rides'' on the DAG:
any validator node can independently and deterministically compute a \emph{commit structure}
to determine the order in which blocks become committed to the backbone by processing the DAG,
without additional communication or historical states outside the DAG.
In previous solutions, deducing the commit structure is somewhat complex,
as the DAG consists of interleaving voting and non-voting block layers.
As we shall see below, the commit structure in \system is very simple:
backbone blocks become committed by themselves,
without the need of voting block layers,
utilizing the invariants from \Complete-Adopt of \bbca.


\paragraph{Separating Data Availability.}

Borrowing from various prevailing systems,
we assume that a separate data availability layer may be used for disseminating data blobs.
First, parallel workers disseminate transactions and generate certificates of availability.
Then, bundles of certificates are assembled into blocks which are totally ordered in the consensus layer.
This way, when a leader proposes a block containing a bundle of certificates, validators trust that their data can be retrieved.
As demonstrated in prior work~\cite{Narwhal}~\cite{yang2022dispersedledger}, offloading data dissemination to worker nodes helps a system to effectively
scale out.
Note that this mechanism alone is not enough to ensure high throughput;
as it has been demonstrated~\cite{mirbft} the throughput capacity of a consensus protocol,
which does not support parallelized proposals,
drops when the number of validator nodes increases, even if they only orders references to data.
However, it can be combined with parallel leader proposals for improving performance and scalability.
The data dissemination and data availability work is left out of scope in this paper
and we note that it can be highly parallelized, and is not considered in the critical path of consensus or counted in latency analysis.

In the rest of this section, we first introduce the \bbca broadcast primitive,
and then we explain how \bbca is used in \system to solve consensus on a DAG\@.

\subsection{\bbca}

The \bbca primitive is an abortable variation of \cbc.
Similar to \cbc, \bbca has a dedicated sender which \primitive-broadcasts messages
and no two correct nodes \primitive-\complete different messages for the same \bbca instance.

Informally, inside \bbca, $f+1$ correct nodes must become locked on a unique message $m$ before \primitive-\complete is possible at any node.
\prim adds an interactive \Complete-Adopt interface to facilitate view changes in partial synchrony which allows nodes to actively probe at any time regardless of their progress in \cbc.
Probing prevents a node from completing the delivery, resulting in one of two possible return values: \textsc{\primitive-adopt} for some locked message $m$
 or \textsc{\primitive-noadopt}, with the following guarantee:
\paragraph{The Complete-Adopt Invariant.}
If any correct node $\bbcacomplete{\bid}{m}{\cert_m}$, then at least $f+1$ correct nodes
return $\bbcaadopt{\bid}{m}{\cert_m}$ upon probing.
In other words, if $f+1$ probes by correct nodes return $\bbcanoadopt{\bid}$, $\bbcacomplete{\bid}{m}{\cert_m}$ never occurs at any correct node.
\\

We implement \bbca via an all-to-all regime \`a la Bracha broadcast~\cite{BrachaT85}.
In this form, \bbca takes $3$ network trips.
To further improve scalability, all-to-all communication may take place over a peer-to-peer gossip transport.
Furthermore, acknowledgements may be batched and aggregated.

The table below summarizes the number of network trips in \bbca compared with other forms of broadcast.
In all forms, the underlying communication may be carried either over all-to-all authenticated channels, or in a linear regime relayed through a sender who aggregates signatures.

\begin{table}[h!]
\small
\centering
\begin{tabular}[width=\textwidth]{r|c|c}
Broadcast scheme & all-to-all & linear \\ \hline \hline
Best Effort Broadcast (BEB) & 1 & 1 \\ \hline
Byzantine Consistent Broadcast (\cbc) & 2 & 3 \\ \hline
\bbca & 3 & 5 \\
\end{tabular}
\vspace{0.2cm}
\caption{Network trips in various broadcast schemes.}
    \label{tab:bcast-latency}
\end{table}

\subsection{\system}

\begin{figure*}[th]
  \centering
  \includegraphics[width=0.85\textwidth]{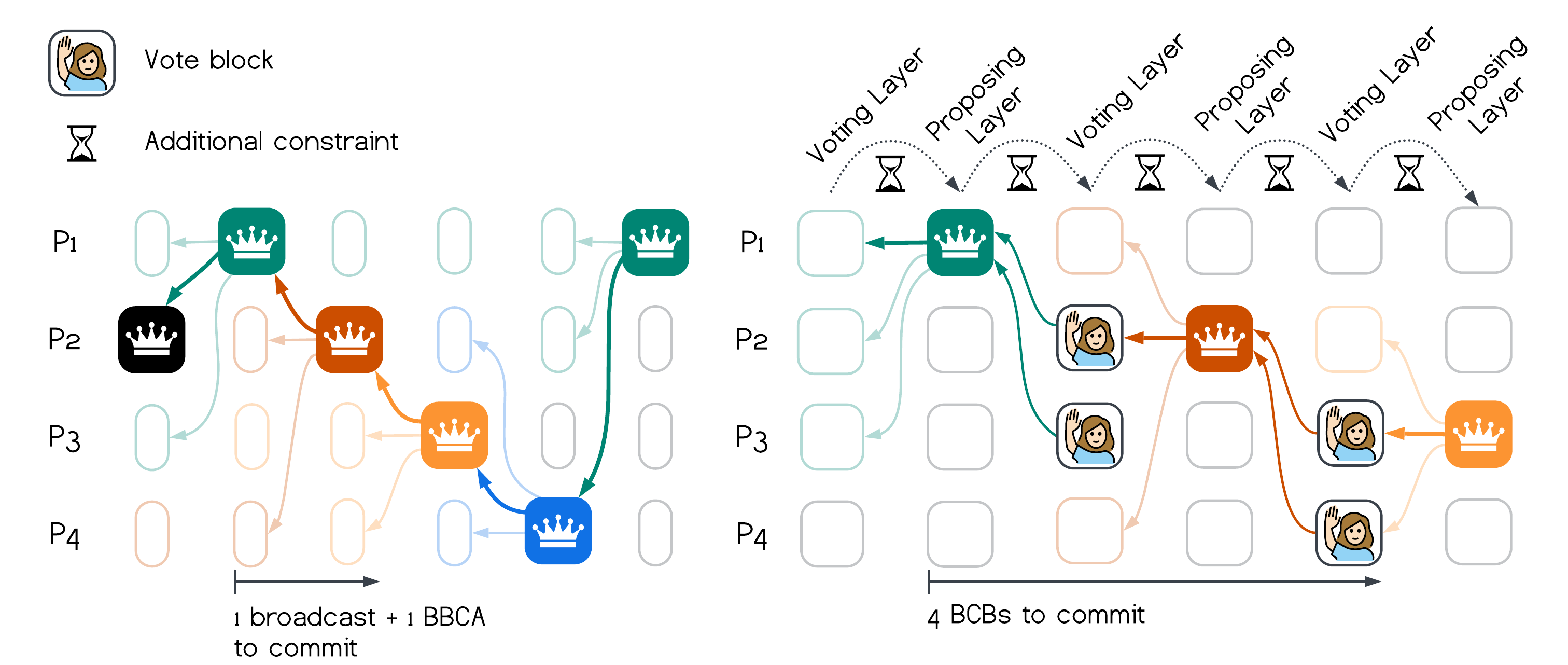}
  \caption{Decision making latency in \system (left) compared with Bullshark (right).}
  \label{fig:BS-N-BBCA}
\end{figure*}

\begin{figure}[h]
\centering
\includegraphics[width=0.9\textwidth]{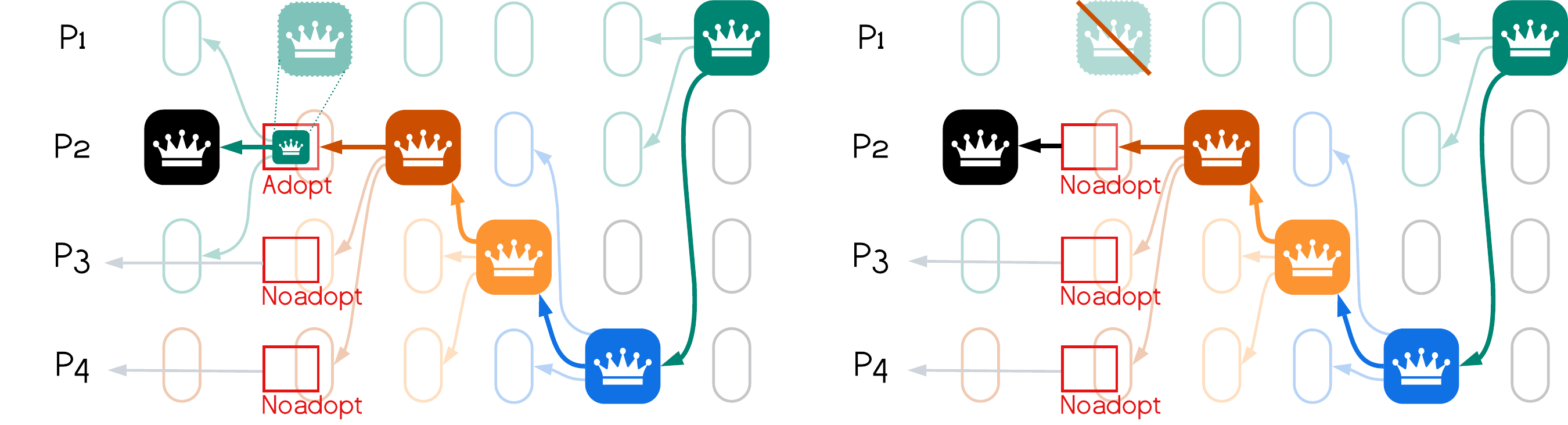}
\caption{When a view timer expires, each node needs to include a \primitive-probe result inside its \newview block. The left DAG depicts an adopted block, and some correct node may have just \primitive-\completed it. The right DAG shows $2f+1$ noadopt responses, proving that no block can be \completed in that view.}
\label{fig:adopt-and-noadopt}
\end{figure}

\system is a view-by-view protocol. Each view reaches a consensus decision in a single \bbca step.
\system forms a DAG of two types of blocks: \leader blocks and \newview blocks. Although there is no artificial layering,
they can be logically grouped by views (as shown on the left of Figure~\ref{fig:BS-N-BBCA}).
Both types of blocks may carry transaction payload.

A transition to view $v$ occurs either upon \primitive-\complete of the \leader block of view $v-1$, or when the new-view timer of view $v-1$ expires.
In the second case,  a node invokes \textsc{\bbca-probe} locally, and then embeds the return value in its \newview block.
A \newview block must include one of the following three items:
a \primitive-\completed \leader block from view $v-1$, an adopted \leader block from probing, or a signed $\textsc{\primitive-noadopt}$ result along with the highest \completed block locally known.
Before entering view $v$, each node other than the leader broadcasts a \newview block, and all nodes start their new-view timers.
The leader of view $v$ \primitive-broadcasts a \leader block. The block refers to the \completed or adopted \leader block of view $v-1$, if available\footnote{As
a technical matter, note that the code in Figure~\ref{code:bbca-chain} encapsulates these references in a \newview block which the leader can embed inside a \leader block.},
otherwise it references $2f+1$ blocks with signed \textsc{\primitive-noadopt}.

The left graph of Figure~\ref{fig:BS-N-BBCA} shows the fast, \emph{commit path}, when the leader has \primitive-\completed the \leader block of the previous view.
Figure~\ref{fig:adopt-and-noadopt} illustrates the \emph{adopt path}, with two different outcomes from nodes' probing results.
The key observation is that the \bbca \Complete-Adopt scheme guarantees that the fast and adopt paths are consistent.%

\bbca allows \leader blocks to become committed by themselves, without any additional votes or DAG layering.
This single-broadcast commit decision accomplishes a substantial reduction in latency to reach consensus, and also significantly simplifies the logic.

Non-\leader blocks are broadcast in parallel using
a Best-Effort Broadcast where a node just directly sends a block to all nodes.  These blocks become committed by being causally referenced by some blocks on the backbone, similar to ``uncle blocks'' in Ethereum.
Since the local DAG of different nodes may evolve at different pace, the \Complete-Adopt feature of \bbca preserves the safety of committed \leader blocks.
\primitive-\completed \leader blocks suffice to uniquely determine which non-\leader blocks to include in the total ordering.

It is worth noting that
previous solutions like Bullshark need non-\leader blocks to vote in the consensus protocol, hence to prevent equivocation, these blocks are broadcast using \cbc.
Because \system gets rid of voting blocks altogether,
there is no need to use \cbc for non-\leader blocks, and this change allows further reduction in latency.
While \cbc is not required for correctness, it may be used for other reasons (e.g.,
to avoid keeping equivocating blocks in the DAG).
Even if \cbc is used for non-\leader blocks, \system exhibits
an overall substantially lower latency than existing DAG-based protocols.

Overall, \system accomplishes a substantial reduction in latency over state-of-art DAG protocols.
Figure~\ref{fig:BS-N-BBCA} illustrates side-by-side the time to commit \leader blocks in \system (left) compared with a DAG-based solution like Bullshark (right).
Hourglasses on the right side denote additional constraints posed on advancing layers by Bullshark.

\section{The \bbca Primitive}
\label{sec:bbca}

\system invokes \primitive{}%
\footnote{Pronounced `Bab-Ka’, and loosely stands for Byzantine Abortable Broadcast with \Complete-Adopt.}
broadcast primitive, an abortable variant of \cbc, for multiple times.
Each \bbca instance has a unique identifier $\bid$ to be totally isolated in its progress.
who can invoke $\Call{\primitive-broadcast}{\bid, m}$, to broadcast message $m$ for instance $\bid$.
Eventually, the node may $\Call{\primitive-\complete}{bid, m, \cert_{m}}$.
We say for short that the node \primitive-\completes $m$ and $\cert_{m}$ is a certificate to convince any node that $m$ has been \completed.
A node may also abort a \bbca instance and probe its state at any time by invoking $\Call{\primitive-probe}{\bid}$.
This invocation returns in two possible ways.
The \bbca instance returns $\bbcaadopt{\bid}{m}{\cert_m}$ for some message $m$ with a certificate $\cert_{m}$,
to certify that some correct node might \complete $m$ and no correct node could \complete a message $m'\neq m$.
Otherwise, \bbca returns $\bbcanoadopt{\bid}$.
In the former case, we say in short that the node \primitive-adopts $m$,
otherwise, we say that it does not \primitive-adopt any message.
\bbca{} maintains the following properties:
\begin{description}
    \item[Validity.] If a correct sender \primitive-broadcasts a message $m$ for instance $\bid$,
        and no correct node probes $\bid$, eventually every correct node \primitive-\completes $m$.

    \item[Consistency.] If a correct node either \primitive-adopts or \primitive-\completes a message $m$
	for instance $\bid$,
    and some other correct node \primitive-\completes message $m'$ for $\bid$,
	then $m=m'$.
    \item[\Complete-Adopt.] If a message $m$ is ever \primitive-\completed by some correct node
	for instance $\bid$,
    then at least $f+1$ correct nodes get $\bbcaadopt{\bid}{m}{\cert_m}$ for $\bid$, if they invoke $\Call{\primitive-probe}{\bid}$.
    That is, if $f+1$ correct nodes get $\bbcanoadopt{\bid}$,
    then no correct node ever \primitive-\completes any message for instance $\bid$.
    \item[Integrity.] If a correct node \primitive-adopts or \primitive-\completes a message $m$ for instance $\bid$
    and the sender $p$ is correct, then $p$ has previously invoked $\Call{\primitive-broadcast}{\bid, m}$.
\end{description}

\bbca is intialized with an external validity predicate $\mathcal{P}$,
such that a correct node only \primitive-\completes $m$ if $\mathcal{P}(m)$ is true.

The pseudocode in Figure~\ref{code:bbca-broadcast} implements \bbca by slightly modifying the Bracha broadcast protocol~\cite{BrachaT85}.
Upon sender's invocation of \textsc{\primitive-broadcast}, it broadcasts $\initMsg{\bid}{m}$.
For a node $p$, upon receiving $\initMsg{\bid}{m}$ from the sender of instance $\bid$,
it checks that validity condition $\mathcal{P}(m)$ holds.
If $m$ is valid, $p$ broadcasts $\echoMsg{\bid}{m}_{\sigma_p}$, where $\sigma_p$ is the signature of $p$ on  $\echoMsg{\bid}{m}$,
and becomes \emph{initialized},
such that it will no longer be initialized again.

Upon receiving an $\textsc{Echo}$ message from $2f+1$ distinct nodes for $\bid$ and $m$ with valid signatures,
$p$ becomes \emph{ready}, \emph{unless} $p$ has already probed instance $\bid$.
Then it broadcasts a signed message $\readyMsg{\bid}{m}_{\sigma_p}$.
The $2f+1$ signatures on $\echoMsg{\bid}{m}$ certify a node has sent a \textsc{ready} message for $m$.

Upon receiving a valid $\textsc{Ready}$ message from $2f+1$ distinct nodes for $\bid$ and $m$,
and, \primitive-\completes $m$ in instance $\bid$.
The $2f+1$ signatures on $\readyMsg{\bid}{m}$ certify a node has completed $m$.

When probed in instance $\bid$, if $p$ has already become ready for $m$, it returns $\bbcaadopt{\bid}{m}{\cert_{m}}$ where $\cert_m$ is a collection of
$2f+1$ signed \textsc{echo} messages. Otherwise, it returns $\bbcanoadopt{\bid}$.

\begin{multicols}{2}
\begin{algorithmic}[1]
    \small
    \Init%
        \InlineComment{$\bid$ is a unique identifier for a BBCA-broadcast instance, it contains $\bid.\sender$ indicating the designated sender of the instance}
        \State$\bbcaState[\bid].\pendingMsgs \gets \{\}$
        \InlineComment{$\pendingMsgs[m].\nEcho \gets \{\}$ and}
        \InlineComment{$\pendingMsgs[m].\nReady \gets \{\}$}
        \InlineComment{when message $m$ is first seen.}
        \State$\forall p \in \allprocesses, \bbcaState[\bid].\receivedEcho[p] \gets \false$
        \State$\forall p \in \allprocesses, \bbcaState[\bid].\receivedReady[p] \gets \false$
        \State$\bbcaState[\bid].\ready \gets \false$
        \State$\bbcaState[\bid].\echo \gets \false$
        \newinbbca{\State$\bbcaState[\bid].\abort \gets \false$}
    \EndInit%

    \Interface{\textsc{BBCA-broadcast}}{\bid, m}
        \State$\Call{becomesInitialized}{\bid, m}$
        \State$\broadcast{\initMsg{\bid}{m}}$
    \EndInterface%

    \newinbbca{%
    \Interface{\textsc{BBCA-probe}}{\bid}
        \State$\letvar{\bstate}{\bbcaState[\bid]}$
        \State$\letvar{\vstates}{\bstate.\pendingMsgs}$
        \If{$\exists m: |\vstates[m].\nEcho| \ge 2f + 1$}
            \State\Return$\bbcaadopt{\bid}{m}{\vstates[m].\nEcho}$
        \Else%
            \State$\bstate.\abort \gets \true$
            \State\Return$\bbcanoadopt{\bid}$
        \EndIf%
    \EndInterface%
    }

    \Procedure{\textsc{becomesInitialized}}{\bid, m}
        \State$\bbcaState[\bid].\echo \gets \true$
        \State$\broadcast{\echoMsg{\bid}{m}_{\sigma_p}}$
    \EndProcedure%

    \Procedure{\textsc{becomesReady}}{\bid, m}
        \If{$\newinbbca{!\bbcaState[\bid].\abort}$}
            \State$\bbcaState[\bid].\ready \gets \true$
            \State$\broadcast{\readyMsg{\bid}{m}}$
        \EndIf%
    \EndProcedure%

    \UponMsg{\initMsg{\bid}{m \neq\bot}}{q}
        \State$\letvar{\bstate}{\bbcaState[\bid]}$
        \If{$\newinbbca{\mathcal{P}(m)} \land q = \bid.\sender  \land (!\bstate.\echo)$}
            \State\Call{becomesInitialized}{$\bid, m$}
        \EndIf%
    \EndUponMsg%

    \UponMsg{\echoMsg{\bid}{m \neq\bot}_{\sigma_q}}{q}
        \State$\letvar{\bstate}{\bbcaState[\bid]}$
        \If{$\newinbbca{\mathcal{P}(m)} \land (!\bstate.\receivedEcho[q])$}
            \State$\bstate.\receivedEcho[q] \gets \true$
            \State$\letvar{\vstate}{\bstate.\pendingMsgs[m]}$
            \State$\vstate.\nEcho \gets \vstate.\nEcho \cup \{\sigma_q\}$
            \If{$!\bstate.\ready \land |\vstate.\nEcho| = 2f + 1$}
                \State\Call{becomesReady}{$\bid, m$}
            \EndIf%
        \EndIf%
    \EndUponMsg%

    \UponMsg{\readyMsg{\bid}{m \neq\bot}_{\sigma_q}}{q}
        \State$\letvar{\bstate}{\bbcaState[\bid]}$
        \If{$\newinbbca{\mathcal{P}(m)} \land (!\bstate.\receivedReady[q])$}
            \State$\bstate.\receivedReady[q] \gets \true$
            \State$\letvar{\vstate}{\bstate.\pendingMsgs[m]}$
            \State$\vstate.\nReady \gets \vstate.\nReady \cup \{\sigma_q\}$
            \If{\removedinbbca{\sout{$!\bstate.\ready \land |\vstate.\nReady| = f + 1$}}}
            \State\removedinbbca{\sout{\Call{becomesReady}{$\bid, m$}}}
            \EndIf%
            \If{$|\vstate.\nReady| = 2f + 1$}
                \State\Call{BBCA-\complete}{$\bid, m, \vstate.\nReady$}
            \EndIf%
        \EndIf%
    \EndUponMsg%

    \captionof{figure}{BBCA-broadcast primitive with a validity predicate $\mathcal{P}$ based on a all-to-all Bracha broadcast protocol. The blue text shows the simple tweak that is added to the vanilla Bracha to enable the \Complete-or-Adopt scheme, whereas red text shows the part that is removed.}%
    \label{code:bbca-broadcast}
\end{algorithmic}
\end{multicols}

For completeness, we mention that another approach for implementation is a linear regime \`a la Cachin et al.~\cite{cachin2000random}. 
In a linear form, \bbca takes $5$ network trips, while the probing logic is similar (return the adopted $m$ or promise not to sign the second acknowledgement):

\begin{enumerate}

\item
The designated sender broadcasts the block to all nodes.

\item
After the validation of the block, nodes respond with a signed acknowledgement over the block's digest.

\item
The sender broadcasts an aggregated/concatenated proof of $2f+1$ signatures.

\item
Nodes respond to the proof from the sender with a signed acknowledgement over the proof.

\item
The sender broadcasts an aggregate/concatenated proof of signatures, and nodes \primitive-\complete it upon receipt.

\end{enumerate}
\paragraph{Relationship to Other Broadcasts}
Bracha's ``all-to-all'' broadcast protocol is not only a \cbc but also a Byzantine Reliable
Broadcast (\rbc), because it additionally guarantees reliable delivery of the broadcast
message. Since \bbca is abortable, reliability is not
guaranteed in the presence of a faulty sender or network delays.
However, when used in the construction of \system as we show later,
reliability is gained for committed blocks through chaining: after seeing a \primitive-\complete certificate for a block, later commits will causally reference the block itself or \primitive-adopt it.

Since we forego reliability for \bbca, we remove the extra trip shown in red
text from the original Bracha protocol (Figure~\ref{code:bbca-broadcast}). Requiring a weaker broadcast primitive
also allows the alternative linear form to preserve linearity, otherwise
all-to-all broadcast would be required on top ``linear'' \cbc in order to guarantee reliability.
Therefore, \bbca makes a weaker, more basic
primitive than \rbc and reduces latency/message-complexity.

\section{The \system Protocol}
\label{sec:bbca-chain}
\subsection{Model}\label{sec:model}
We assume a set of $n$ nodes which communicate via point-to-point authenticated channels
such that up to $f=\frac{n-1}{3}$ of them may fail arbitrarily.
We assume partially synchronous communication such that network delays are unbounded until some unknown Global Stabilization Time (GST).
After GST, network delays are bounded by constant $\Delta$. But because GST is unknown, the safety of the protocol cannot rely on $\Delta$.

\subsection{Best-Effort Broadcast with Causal-Ordering}
Throughput the paper, we use the term \emph{broadcast} to denote a
\emph{best-effort broadcast} where a node sends the block to all other nodes.
Moreover, we say that block $B_1$ \emph{causally precedes} $B_2$, or that $B_1$ is a \emph{causal predecessor} of $B_2$,
to denote that $B_2$ includes a \emph{causal reference} to block $B_1$.
Such reference could be implemented simply as a hash of $B_1$.
Causal ordering guarantees that if $B_2$ causally references $B_1$, then no correct node delivers $B_2$ before $B_1$.
Throughout the paper we always assume that nodes causally reference their own previously broadcast block.

\subsection{Consensus}\label{sec:blockchain}
In a consensus protocol, each correct node may take some payload (e.g., transactions), seal it in a block, and send it to the network. A \textbf{commit} statement is triggered to sequence blocks
in their final order, with the following properties:
\begin{description}
    \item[Consistency.] For the sequences of \textbf{committed} blocks by any two nodes, one has to be the prefix of another.
    \item[Chain growth.] The sequence of \textbf{committed} blocks keeps growing for all correct nodes.
    \item[Censorship resistance.] After GST, if a correct node sends a block $b$ to the network, then $b$ is eventually \textbf{committed}.
\end{description}

\system commits blocks in \emph{views} and in each view $v$,
one node acts as \emph{leader},
evaluated by each node via a $\Call{getProposer}{v}$ function.
For simplicity, throughout the paper we assume a round-robin rotation,
though more elaborate and/or adaptive schemes may be used.
Nodes enter a new view with a \emph{view synchronization} protocol, which we discuss in more detail in Section~\ref{sec:practical}.
Each node, upon entering a new view,
starts a view timer and starts
participating in a new \bbca instance (identified as $\bid = \langle \Call{getProposer}{v}, v \rangle$).
Nodes other than the leader broadcast a \emph{\newview} block.
The leader, however, proposes a \emph{\leader block} $B_v$ by invoking $\Call{\primitive-broadcast}{\bid, B_v}$.

In the common case, a node completes view $v$ and is ready to enter view $v+1$ upon it \primitive-\completes the \leader block of view $v$. The block that has completed
gets marked as the final block for $v$ and the protocol guarantees that every correct node eventually also commits this block.
In this case, the \newview block for view $v+1$ includes
(1) a causal reference to the completed \leader block $B_{v}$,
(2) the \primitive-\complete certificate $\cert_{B_v}$ proving that $p$ completed $B_{v}$,
(3) the view number $v$, and (optionally) a payload of pending transactions and causal references to help commit some additional previous blocks.
For a node that is the leader of $v+1$, instead of broadcasting this \newview{} block, it just embeds (or causally references to) it when it propose in view $v+1$.

If at some node $p$ the timer for view $v$ expires before a \primitive-\complete of \leader block for view $v$,
then $p$ \emph{probes} the \bbca instance of view $v$ and finishes view $v$.
This forces the \bbca instance to stop sending ready messages internally.
If there exists a node that could \primitive-\complete a \leader block $B_v$,
then \bbca returns $\langle\textsc{\primitive-adopt}, B_v, \cert_{B_v}\rangle$,
for at least $f+1$ correct nodes
where $\cert_{B_v}$ is the corresponding adopt certificate.
In this case, the \newview block for view $v+1$ includes
(1) a causal reference to the adopted \leader block $B_{v}$,
(2) the adopt certificate $\cert_{B_v}$ proving that $p$ adopted $B_{v}$ (i.e., no correct node could complete any other block),
(3) the view number $v$, and like the previous case, optional payload and causal references.

For the rest of the nodes, \bbca may return $\langle\textsc{\primitive-noadopt}\rangle$.
In this case the \newview block which $p$ broadcasts includes
(1) a signed tuple $\langle\textsc{noadopt},v\rangle_{\sigma_p}$,
(2) the view number $v$, and like the previous cases, optional payload and causal references.

The \leader block for view $v+1$
includes
(1) a causal reference to $1$ or $2f+1$ \newview blocks, depending on how the leader enters the view,
(2) the new view number $v+1$, and like \newview blocks, possibly some additional payload.
In the common case, the leader includes a causal reference to a \newview block which references the \leader block of view $v$
which the leader of view $v$ has committed.
Notice that in this case the leader can \emph{embed} its own \newview block,
and, hence, does not have to wait for other \newview blocks to arrive.
Otherwise, the leader waits for \newview blocks of other nodes to justify what happened in view $v$.
If the leader receives a \newview block $B_n$ with a valid adopt certificate $\cert_{B_v}$ for the \leader block $B_v$ of view $v$,
the leader includes a causal reference to $B_n$.
Otherwise, the leader waits for and causally references $2f+1$ \newview blocks from distinct nodes with a valid signature on the tuple $\langle\textsc{noadopt}, v\rangle$.
These constitute a justification that no node could commit a \leader block in view $v$ and thus it is safe to skip.

Each node maintains local state for \emph{finalized} views,
i.e., each view $v$ for which the node has a valid complete certificate for a \leader block,
a valid adopt certificate for a \leader block causally followed by a completed \leader block,
or $2f+1$ distinct valid signatures on $\langle\textsc{noadopt}, v\rangle$  causally followed by a completed \leader block.
Upon receiving a \newview block, including their own, with a causal reference to a \leader block $B_v$ for view $v$ and a valid complete certificate for $B_v$,
a node $p$ finalizes view $v$ with $B_v$.
Moreover, $p$ recursively finalizes the non-finalized views in $B_v$'s causal history as follows:
a view $v'$ is finalized with $B_{v'}$ if there exists a complete or adopt certificate for $B_{v'}$,
otherwise $v'$ is finalized with a special $\textsc{no-op}$ value to indicate that view $v'$ should simply be skipped.

Additionally, the node maintains a pointer for the highest committed view $\lastCommitted$.
Upon finalizing some view $v$, $p$ commits \leader blocks for all contiguous finalized views, starting from $\lastCommitted + 1$, while skipping views finalized with $\textsc{no-op}$.
This ensures that nodes commit \leader blocks in increasing view numbers with consistent committed prefixes.
Along each \leader block, $p$ commits to the total order of all remaining causally referenced blocks including
\newview blocks according
to some deterministic rule.
Note that by the causality of broadcast, a node that commits the \leader block already has all the blocks referenced by the \leader block available.

The \system protocol is presented as pseudocode in Figure~\ref{code:bbca-chain}.

\begin{multicols}{2}
\begin{algorithmic}[1]
    \small
    \Init%
        \State$\forall v \in \mathbb{N}, \newviewblocks[v] \gets \bot$
        \InlineComment{$\newviewblocks[v]$ holds the set of received new-view blocks in view $v$, $\newviewblocks[v][q]$ is the block from $q$}
        \State$\forall v \in \mathbb{N}, \finalized[v] \gets \bot$
        \InlineComment{$\finalized[v]$ holds the finalized block for view $v$, a continuously committed prefix is the committed chain.}
        \State$\lastCommitted \gets 0$
        \State$\newviewblocks[0] \gets \newViewBlock($\\
        \hspace{2em}$\blockref = B_\mathit{genesis}, \mathit{\completeCert} = \cert_0, \view = 0)$
        \State$\view \gets 0$
        \State$\viewTimer.\mathit{init}()$
        \InlineComment{start the view advancing loop}
        \State$\Call{enterView}{1}$
    \EndInit%
    \Procedure{\textsc{enterView}}{v}
        \State$\view \gets v$
        \State$\viewTimer.\mathit{restart}()$
        \State$\letvar{\mathit{proposer}}{\Call{getProposer}{\view}}$
        \State{}initialize BBCA instance with $\bid = \langle\mathit{proposer}, \view\rangle$
	    \InlineComment{Nothing for non-leader to do until some event received}
        \State\ifthen{$\mathit{proposer} \neq p$}{\Return}
            \State$\letvar{C}{\newviewblocks[\view - 1]}$
        \State$\letvar{\mathit{\complete}}{\exists q: C[q].\mathit{\completeCert} \neq \bot\,\,\land}$\\
            \hspace{9em}$C[q].\blockref.\view = \view - 1$
            \State$\letvar{\mathit{adopt}}{\exists q: C[q].\mathit{adoptCert} \neq \bot}$
            \State$\letvar{\mathit{noadopt}}{\exists Q: |Q| = 2f + 1\,\,\land}$\\
            \hspace{4em}$(\forall r \in Q: C[r].\sig = \langle \textsc{noadopt}, \view - 1 \rangle_{\sigma_r})$
            \State{}\textbf{wait until} $\mathit{\complete} \lor \mathit{adopt} \lor \mathit{noadopt} = \true$
            \If{$\mathit{\complete}$}
                \State$\prev \gets \langle \committed = C[q] \rangle$
            \ElsIf{$\mathit{adopt}$}
                \State$\prev \gets \langle \adopted = C[q] \rangle$
            \Else%
                \State$\prev \gets \langle \skipped = \{ C[r]: \forall r \in Q \}  \rangle$
            \EndIf%
        \State$B \gets \mathit{backboneBlock}(\blockref = \prev, \view = \view)$
        \State$\Call{BBCA-broadcast}{\langle \sender = p, \view = \view \rangle, B}$
    \EndProcedure%

    \Upon{\Call{\prim-complete}{\bid, B, \cert_B}}
        \State$\Call{tryCommit}{B}$
        \State\ifthen{$B.\view < \view$}{\Return}
        \State$\letvar{B_n}{\newViewBlock(}$\\
        \hspace{2em}$\blockref = B, \mathit{\completeCert} = \cert_B, \view = B.\view)$
            \If{$\Call{getProposer}{B.\view + 1} = p$}
                \InlineComment{$p$ is the proposer for the next view, thus}
                \InlineComment{\textsc{enterView} will embed $B_n$ in the proposal.}
                \State$\newviewblocks[B_n.\view][p] \gets B_n$
            \Else%
                \InlineComment{Otherwise, broadcast the \newview block.}
                \State$\textsc{broadcast}(B_n)$
            \EndIf%
            \State$\Call{enterView}{B.\view + 1}$
    \EndUpon%
    \UponCondition{\viewTimer.\mathit{elapsed} > \mathit{T_{max}}}
        \If{$\langle \textsc{BBCA-adopt}, \bid, B, \cert_{B} \rangle = \Call{BBCA-probe}{\bid}$}
        \State$\textsc{broadcast}(\newViewBlock($\\
        \hspace{4em}$\blockref = B, \mathit{adoptCert} = \cert_B, \view = \view))$
        \Else%
            \State$\textsc{broadcast}(\newViewBlock($\\
            \hspace{4em}$\sig = \langle \textsc{noadopt}, \view \rangle_{\sigma_p},$\\
            \hspace{4em}$\view = \view))$
            \InlineComment{Implicitly causally refers to the last completed or adopted block in the causal history of $B$.}
        \EndIf%
        \State$\Call{enterView}{\view + 1}$
    \EndUponCondition%

    \UponReceivingX{newViewBlock}{B_n}{q}
        \If{$\newviewblocks[B_n.\view][q] = \bot$}
        \State$\newviewblocks[B_n.\view][q] \gets B_n$
        \State\ifthen{$B_n.\mathit{\completeCert} \neq \bot$}{$\Call{tryCommit}{B_n.\blockref}$}
        \EndIf%
    \EndUponReceivingX%

    \Procedure{\textsc{tryCommit}}{B}
        \State$\Call{finalize}{B}$
        \While{\finalized[\lastCommitted+1] \neq\bot}
            \State$\lastCommitted \gets \lastCommitted + 1$
            \If{$\finalized[\lastCommitted] \neq \textsc{no-op}$}
            \State$\textbf{commit}\,\finalized[\lastCommitted]$
            \EndIf%
        \EndWhile%
    \EndProcedure%

    \Procedure{\textsc{finalize}}{B}
    \State\ifthen{$\finalized[B.\view] \neq \bot$}{\Return}
    \State$\finalized[B.\view] \gets {B}$

    \State$B_{\max} \gets$ the causal predecessor of $B$ of the highest view with an adopt or complete certificate
        \ForAll{i}{[B_{max}.\view+1,{}B.\view)}
        	\State$\finalized[i] \gets \textsc{no-op}$
        \EndForAll
    \State$\Call{finalize}{B_{\max}}$
    \EndProcedure%

    \captionof{figure}{BBCA-Chain operational logic for node $p$. Block validation is omitted in the code for clarity. $\mathit{backboneBlock}()$
        and $\mathit{newViewBlock}()$ construct backbone and new-view blocks respectively. During the block construction, actual payload (e.g., transactions from the mempool) and causal references to the uncommitted frontier can be added to causally commit useful transactions according to the consistent ordering on the view-by-view backbone determined by $\textbf{commit}$.}%
    \label{code:bbca-chain}
\end{algorithmic}
\end{multicols}

\section{Further Discussion}
\label{sec:practical}

\paragraph{Simplifying further.}
We can change the protocol to avoid requiring nodes to broadcast \newview blocks.
Adopting the principles of HotStuff-2~\cite{malkhi2023hotstuff},
nodes will send \emph{\newview messages} directly to the leader of the next view.
The view-synchronization would be handled using some protocol outside the DAG (see the related works in Section~\ref{sec:related}).

Briefly, this modification works as follows.
The leader of view $v$ waits for the earliest of the following conditions in order to form a proposal:

\begin{itemize}
\item a commit certificate for $B_{v-1}$ is received.
\item an adopt certificate has been received.
\item \newview messages on view $v-1$ have been received from all nodes.
\item the view $v-1$ timer has expired, potentially with an additional time to allow one network trip.
\end{itemize}

Other nodes guard the safety of \leader{} proposals by comparing with their own highest certified \leader block,
instead of looking at the proposal causal history. That is,
the predicate $\mathcal{P}(B_v)$ nodes use before they accept a \leader block in \primitive{} is that
the block references a \leader block (either as commit or an adopt certificate) at least as high as the node's highest adopt/commit certificate.

Briefly, after GST,
in case the leader of view $v$ does not obtain a certificate for a \leader block of view $v-1$, then it is guaranteed that no correct node holds a commit or adopt certificate for view $v-1$.
Before GST, the leader might not receive in time the highest \leader block certificate in the system.
In case there is any node with a higher \leader block certificate than proposed, then this node may indeed reject the leader proposal because it will not pass the $\mathcal{P}(B_v)$ predicate.
This is ok, because safety is still ensured by BBCA, while liveness can only be guaranteed after GST.

\paragraph{Network utilization.}
As discussed in Section~\ref{sec:bbca-chain} \leader blocks and \newview blocks
may include a payload, i.e., transactions or a bundle of certificates for available transactions, assuming a separate data dissemination and availability layer.
However, \system design allows nodes to broadcast \regular blocks (with payload) independently in parallel, at any time, and orthogonally to view progression.
Therefore, unlike protocols based on layered DAGs, nodes with high load and high bandwidth capacity  can broadcast \regular blocks more frequently.
Moreover, nodes can continue broadcasting \regular blocks, even when view progression is slow, in the presence of a network partition or a slow/Byzantine leader.
This allows nodes to stay busy continuously and truly saturate the network capacity.

\paragraph{View synchronization.}
In BFT protocols, view synchronization ensures that all correct nodes enter a new view within bounded time.
There exist several recent advancements in literature which could be used by \system~\cite{Civit2022ByzantineCI,LewisPye2022QuadraticWM}.
However, we can enhance \system itself to implement view synchonization with few simple rules.
In the common case, a node completes view $v$ and enters view $v+1$ upon it \primitive-\completes the \leader block of view $v$.
Upon receiving a \newview block $B_n$ for view $v+1$, which references a completed \leader block, the node completes view $v$ and enters view $v+1$.
Similarly, if $B_n$ references an adopted \leader block, the node adopts the \leader block, and enters view $v+1$.
Otherwise, a node that encounters the timeout probes \bbca and gets \textsc{\primitive-noadopt}, does not enter view $v+1$ unless it receives $2f+1$ \newview blocks.
Additionally, a node that receives $f+1$ \newview blocks for view $v+1$ with a valid $\langle \textsc{noadopt}, v+1\rangle$ signature, probes its \bbca instance for view $v$, completes view $v$ and enters view $v+1$.
In the last case, $f+1$ blocks are necessary to ensure that adversaries cannot force a correct node to complete a view.

\section{\bbca Implementation Correctness}
In this section  we prove that the specification of \bbca in Algorithm~\ref{code:bbca-broadcast} satisfies \bbca properties described in Section~\ref{sec:bbca}.

\begin{lemma}[Validty]
\label{lem:bbca-validity}
    If a correct sender \primitive-broadcasts a message $m$ for instance $\bid$,
    and no correct node probes $\bid$, eventually every correct node \primitive-\completes $m$.
\end{lemma}

\begin{proof}
    Because the sender is correct, all correct nodes broadcast an $\textsc{Echo}$ message for $m$.
    All correct nodes then receives at least $n-f=2f+1$ $\textsc{Echo}$ messages.
    As no correct node invokes \textsc{\primitive-probe},
    all correct nodes will broadcast a $\textsc{Ready}$ message for $m$.
    Therefore, all correct nodes receive at least $2f+1$ $\textsc{Ready}$ messages and \primitive-\complete $m$.
\end{proof}

\begin{lemma}[Consistency]
\label{lem:bbca-conistency}
    If a correct node either \primitive-adopts or \primitive-\completes a message $m$
    for instance $\bid$,
    and some other correct node \primitive-\completes message $m'$ for $\bid$,
    then $m=m'$.
\end{lemma}

\begin{proof}
    For either \primitive-adopt or \primitive-\complete case, a correct node $p$ needs to receive at least $2f+1$  $\textsc{Echo}$ messages.
    Let's assume that some other node $q$ \primitive-\completes some message $m'\neq m$ in the instance of $\bid$.
    Then $q$ has also received at least $2f+1$  $\textsc{Echo}$ messages.
    Since $n=3f+1$, among $2(2f+1) = 4f + 2$ messages, there are $f+1$ sent by the same set of nodes to both $p$ and $q$.
    Thus at least one correct node has sent to both $p$ and $q$, leading to a contradiction that a correct node can only echo
    the same $m$.
\end{proof}

\begin{lemma}[\Complete-Adopt]
\label{lem:bbca-complete-adopt}
    If a message $m$ is ever \primitive-\completed by some correct node
    for instance $\bid$,
    then at least $f+1$ correct nodes get $\bbcaadopt{\bid}{m}{\cert_m}$ for $\bid$, if they invoke $\Call{\primitive-probe}{\bid}$.
    That is, if $f+1$ correct nodes return $\bbcanoadopt{\bid}$,
    then no correct node ever \primitive-\completes any message for instance $\bid$.
\end{lemma}

\begin{proof}
    If some correct node \primitive-\completes a message $m$, then it has received a $\textsc{Ready}$ message from at least $2f+1$ nodes. Among them, at least $f+1$ are correct nodes, who have an adopt certificate (consists of $2f+1$ signed \textsc{Echo} messages) for $m$.
    Therefore they will return \textsc{\primitive-adopt} upon probing.

    For the second part of the lemma, assume that some correct node $p$ \primitive-\completes a message $m$.
    $p$ has received $2f+1$ $\textsc{Ready}$ messages, of which $f+1$ are from the correct nodes.
    For the group of $f+1$ nodes which return \textsc{\primitive-noadopt}, there is at least one node that also has sent the aforementioned \textsc{Ready} message, which is contradictory because a node cannot both send a ready message and returns \textsc{\primitive-noadopt}, in any order (in the pseudocode, probing will mark $\abort \gets \true$, preventing the node from sending \textsc{Ready}; and after sending \textsc{Ready}, probing will always return the adopted message).
\end{proof}

\begin{lemma}[Integrity]
\label{lem:bbca-integrity}
    If a correct node \primitive-adopts or \primitive-\completes a message $m$ for instance $\bid$
    and the sender $p$ is correct, then $p$ has previously invoked $\Call{\primitive-broadcast}{\bid, m}$.
\end{lemma}

\begin{proof}
    Correct nodes only broadcast $\textsc{Echo}$ messages for an $\textsc{Init}$ message from the instance sender, so
    a completed message must have been initially sent in \textsc{Init} message, sent by $p$ upon the invocation of \textsc{\primitive-broadcast}.
\end{proof}
\section{\system Correctness}

\subsection{Safety}
To prove consistency, we need to show that
correct nodes commit the same blocks in the same order.
It suffices to show that correct nodes commit \leader blocks in the same order,
as the order of non \leader blocks is determined by the first \leader block that references them
and some arbitrary deterministic rule.

\begin{lemma}
\label{lem:unique-cert}
There can be at most one unique $B$ per view obtaining a \complete certificate and/or an adopt certificate (for all correct nodes).
\end{lemma}

\begin{proof}
The lemma follows directly from \bbca consistency, as in each view correct nodes initialize a single \bbca instance.
\end{proof}

\begin{lemma}
\label{lem:causal-cert}
Suppose a node $p$ commits a \leader block $B_v$ for view $v$.
Then there exists a view $v'$ where $v' \geq v$, such that $p$ has \primitive-\completed a block $B_{v'}$ at view $v'$ and
such that there is a causal chain of \complete certificates or adopt certificates of \leader blocks from $B_{v'}$ back to $B_v$.
\end{lemma}

\begin{proof}
Nodes process commit blocks in strictly increasing views, iterating back to their last committed view.
The lemma follows immediately from the code:
procedure \textsc{tryCommit} is invoked only upon \primitive-\complete, and iterates backward through blocks in the $\finalized$ array,
 each of which has a \complete or an adopt certificate.
\end{proof}

\begin{lemma}
\label{lem:unique-commit}
If two correct nodes $p$, $q$ commit \leader blocks $B_1$, $B_2$ respectively for view $v$, then $B_1 = B_2$.
\end{lemma}

\begin{proof}
For the Lemma statement to hold, $p$ and $q$ have finalized view $v$ with $B_1$ and $B_2$ respectively.
Therefore, both $B_1$ and $B_2$ received a \complete and/or an adopt certificate for view $v$, and by
Lemma~\ref{lem:unique-cert} $B1=B2$.
\end{proof}

\begin{lemma}[Consistency]
If a correct node $p$ commits $B_v$ for some view $v$, and another
correct node $q$ commits any $C_{w}$ for view $w \geq v$, then $q$ commits $B_{v}$ for view $v$.
\label{sem:safety2}
\end{lemma}

\begin{proof}
Let $v'$ be the lowest view where $v' \geq v$ for which $p$ \primitive-\completes some block $B_{v'}$ ($v'$ exists
by Lemma~\ref{lem:causal-cert}).
Likewise, let $w'$ be the lowest view where $w' \geq v$ for which $q$ \primitive-\completes some block $C_{w'}$ (likewise, $w'$ exists).
Without loss of generality, assume $w' \geq v'$.

If $w' = v'$,
then and by Lemma~\ref{lem:unique-commit}, $C_{w'} = B_{v'}$.
Otherwise, $w' > v'$.

When $q$ iterates from view $w'$ back to view $v$, it goes through view $v'$.
By assumption, $p$ \primitive-\completed $B_{v'}$ and
$q$ did not obtain a \complete certificate for view $v'$.
Therefore, $q$ either adopts $B_{v'}$ or skips view $v'$.
However,
by the \Complete-Adopt property of \bbca,
$q$ cannot skip with $2f+1$ $\langle\textsc{\primitive-noadopt}\rangle$ for view $v'$,
since $p$ has \primitive-\completed a block in view $v'$.
Therefore, $q$ commits some \leader block in view $v'$ via an adopt certificate,
and by Lemma~\ref{lem:unique-cert} this block must be $B_{v'}$.

In both scenarios, by Lemma~\ref{lem:unique-commit}, both $p$ and $q$ iterate from $B_{v'}$ back to $B_v$ through an identical sequences of committed \leader blocks.
\end{proof}

\subsection{Liveness}
We assume a view synchronization protocol which guarantees that, after GST, all correct nodes enter view $v$ within bounded time $\Delta_{sync}$.
Section~\ref{sec:view-sync} shows \system can satisfy this property
without incurring extra communication, using the DAG causal ordering of blocks by all nodes.

A view-timer $\Gamma$ with duration $\Delta_\Gamma$ is picked such that after GST, if $2f+1$ correct nodes enter view $v$ at time $t$ then
$\Gamma$ suffices for the leader of view $v$ to

(i) Enter view $v$,

(ii) Broadcast a \leader block via \bbca, and reach $\Call{\bbca-\complete}{}$ at all $2f+1$ correct nodes.

I.e. $\Gamma\ge \Delta_{sync} + 3\Delta$, where $\Delta$ the maximum network delay after GST.

\begin{lemma}
\label{lem:live-1}
After GST, if the leader of view $v$ is correct, $\Call{\bbca-\complete}{}$ is reached in view $v$ by all correct nodes.
\end{lemma}

\begin{proof}
Let as assume that we are after GST and the first correct node $p$ to ender view $v$ does so at time $t$.
By Lemma~\ref{lem:view-sync}, all correct nodes, including the leader, enter view $v$ the latest at time $t+\Delta_{sync}$.
Therefore the leader broadcasts a block $B_v$ the latest at time $t+\Delta_{sync}$,
that will get accepted inside \bbca (we assume the leader is correct), the latest at $t+\Delta_{sync}+\Delta$ by all correct nodes.
The latest at $t+\Delta_{sync}+2\Delta$ all correct nodes will have received $2f+1$ $\textsc{Echo}$ messages and
the latest at $t+\Delta_{sync}+3\Delta$ all correct nodes will have \bbca-\completed $B_v$, i.e, before $\Gamma$ expires for $p$.
\end{proof}

\begin{lemma}
\label{lem:live-4}
For every view $v$, a correct node eventually enters the view.
Furthermore, upon entering view $v$, the node has received one of the following from view $v-1$:
(i) a commit certificate of the \leader block of view $v-1$,
(ii) an adopt certificate of the \leader block of view $v-1$,
(iii) $2f+1$ signed noadopt of the \leader block of view $v-1$.
\end{lemma}

\begin{proof}
We prove this by induction.
After entering view $v$,
each node broadcasts a \newview block for view $v-1$
either upon receiving messages from others (including possibly the \leader block for view $v-1$), or its own view $v-1$ timer expires.
Since eventually all correct nodes send \newview blocks, a node $p$ enters view $v$ when one of the following conditions holds:

\begin{itemize}
\item $p$ received $2f+1$ \newview messages for view $v-1$ carrying (signed) $\textsc{\primitive-noadopt}$ (for view $v-1$).
\item $p$ received a commit-certificate for view $v-1$.
\item $p$ received a \newview message carrying an adopt-certificate for view $v-1$.
\end{itemize}
\end{proof}

\begin{lemma}
\label{lem:live-2}
For every view $v$ a correct node either commits a block, or skips the view.
\end{lemma}

\begin{proof}
By Lemma~\ref{lem:live-4}, eventually every correct node enters $v$ having received a ``skip'' certificate ($2f+1$ noadopt's), a commit certificate, or an adopt certificate.
By Lemma~\ref{lem:live-1}, after GST eventually there is a view $v' \geq v$ whose leader is correct and
all nodes \bbca-\complete the \leader block of $v'$.
A node then invokes $\Call{tryCommit}{}$ which iterates backward starting at the committed block of view $v'$.
Starting with the block that becomes committed in view $v'$, each committed block $B$ encountered in the loop determines a block $B_{prev}$ that immediately precedes $B$ in the sequenced $\backbone$:

\begin{itemize}
\item
If $B.\blockref$ references a commit or an adopt certificate, then $B_{prev}$ is the certified block.
\item
Otherwise, $B_{prev}$ is the block whose view is maximal among the commit and adopt certificates which are referenced by blocks in $B.\blockref.\skipped$.
\end{itemize}

Then, all the views between $B_{prev}$ and $B$ becomes finalized as skipped;
$B_{prev}$ becomes committed and the next loop continues iterating backward from it.

The loop completes when all non-finalized views $\leq v'$ become finalized.
Therefore, if view $v$ hasn't been finalized already, it will become finalized either as skipped or committed.
\end{proof}

\begin{lemma}[Chain growth]
\label{lem:growth}
The chain of blocks grows infinitely, i.e., eventually, all correct nodes commit a \leader block for some view $v$.
\end{lemma}

\begin{proof}
By Lemma~\ref{lem:live-1} and the assumption that $\textsc{GetProposer}(v)$ returns a correct leader for some view $v$ after GST,
all correct nodes \bbca-\complete some \leader block $B_v$ for view $v$.
We therefore need to show that eventually $B_v$ will also be committed.
By Lemma~\ref{lem:live-2}, for every view $v'<v$ either some block is committed, or the view is skipped.
Therefore, $v$ will be recursively committed.
\end{proof}

\begin{lemma}[Censorship resistance]
\label{lem:censorship-resistance}
After GST, if a correct node broadcasts a block $b$, then the block will be eventually committed.
\end{lemma}

\begin{proof}
If the block is a \leader block, then by Lemma~\ref{lem:live-1}, $b$ gets \bbca-\completed for some view $v$.
By Lemma~\ref{lem:live-2}, for all views $v'<v$ eventually, either some block gets committed or the view is skipped and, therefore, eventually $b$ gets committed.
Otherwise the block will be committed as a causal reference of a \leader block.
\end{proof}

\subsection{View Synchronization}
\label{sec:view-sync}

\begin{lemma}[View Synchronization]
\label{lem:view-sync}
Let $p$ be the first correct node enters view $v+1$ at time $t$ after GST.
Then all correct nodes view $v+1$  the latest at time  $t+\Delta$, where $\Delta$ is the maximum network delay.
\end{lemma}

\begin{proof}
We examine exhaustively all ways in which $p$ may enter view $v+1$.
\begin{enumerate}
	\item $p$ \primitive-\completes the \leader block of view $v$. In this case $p$ broadcasts a \newview block for $v+1$, which all correct nodes receive the latest after $\Delta$. This allows them to enter view $v+1$ themselves.
	\item $p$ receives a \newview block for view $v+1$ which references a completed \leader block for view $v$. As in the previous case, $p$ broadcasts a \newview block for $v+1$, and all correct nodes enter view $v+1$ the latest after $\Delta$.
	\item $p$ receives a \newview block for view $v+1$ which references an adopted \leader block for view $v$. Again, $p$ broadcasts a \newview block for $v+1$, and all correct nodes enter view $v+1$ the latest after $\Delta$.
	\item $p$ probes \bbca (either due to a timeout, or after receiving $f+1$ \newview messages with \textsc{noadopt}) which returns \textsc{\primitive-adopt}. In this case $p$ broadcasts a \newview block for $v+1$ with an adopt certificate, which all correct nodes receive the latest after $\Delta$ and enter view $v+1$ themselves.
	\item $p$ probes \bbca (either due to a timeout, or after receiving $f+1$ \newview messages with \textsc{noadopt}) which returns \textsc{\primitive-noadopt}.
	By the assumption that $p$ is the first correct node who enters view $v+1$, $p$ has received $2f+1$ \newview messages with a $\langle \textsc{noadopt}, v+1\rangle$ signature.
	At least $f+1$ of them are from correct nodes.
	Let's assume that the $f+1^{th}$  correct node who sent a \newview message with $\langle \textsc{noadopt}, v+1\rangle$ did so at time $\tau$.
	Therefore, all corrects nodes receive $f+1$ \newview messages by time at most $\tau+\Delta$ and send a \newview message themselves.
	Therefore, all correct nodes receive in total $2f+1$ messages by time $\tau + 2\Delta$ and enter view $v+1$, that is at most $\Delta$ after $p$.
	\end{enumerate}
\end{proof}

\section{Related Work}
\label{sec:related}

The focus of our work is solving consensus and state-machine-replication (SMR) while allowing underlying parallel communication for high throughput.

We therefore review SMR solutions along a spectrum: on one extreme are traditional leader-based SMR protocols, in the middle are solutions that leverage causality to parallelize certain aspects, and at the far extreme are leaderless/multi-leader solutions that commit proposals at bulk.

PBFT~\cite{PBFT}, a landmark in BFT solutions introduced two decades ago, emphasizes optimistically low latency.
It established the view-by-view ``recipe'' where the leader of a new-view submits a proposal with a \emph{justification proof}.
This approach for justifying a new leader proposal after a view-change is the foundation of all protocols in the PBFT family,
including FaB~\cite{FaB-06}, Zyzzyva~\cite{zyzzyva-07}, Aardvark~\cite{Aardvark}, and SBFT~\cite{SBFT-19}.
Notably, the last builds on the pioneering works of Cachin et. al~\cite{cachin2001secure} and Reiter~\cite{rampart-94}, to linearize the common case message complexity by employing signature aggregation.
Still, the view-change justification proof is complex to code and incurs quadratic communication complexity.
Tendermint~\cite{buchman2016tendermint} introduced a simpler view-change sub-protocol than PBFT, later adopted in Casper~\cite{casper}.
HotStuff~\cite{HotStuff-19} harnesses and enhances the simple Tendermint view-change via an extra phase;
several advances to HotStuff~\cite{jalalzai2020fast,diembft-v4,gelashvili2022jolteon,mscfcl,wendy} eventually led to
HotStuff-2~\cite{malkhi2023hotstuff} which successfully removes the requirement of an extra phase.

Follow-up research works~\cite{crain2021red,yang2022dispersedledger,RCC,mirbft,smrmadesimple,LevAri2019FairLedgerAF}
focused on improving throughput by enabling multiple nodes to act as leaders in parallel, and therefore, alleviate the computational and communication single-leader bottleneck of previously mentioned works.
Still, those protocols require similar view-change mechanisms with their single-leader counterparts, resulting in periods with underutilized network resources.
This issue was addressed by DAG-based solutions which employ causal ordering and pipelining to decouple consistent block dissemination from ordering.

One class of DAG-based algorithms are asynchronous and leaderless and include DAG-based solutions that date back to the 1990's as well as recent protocols for blockchains~\cite{dolev1993early,moser1999byzantine,baird2016swirlds,DAG-rider,Narwhal}.
Recently, a class of DAG-based algorithms emerged that borrow the leader-based approach of algorithms like PBFT/HotStuff and optimize performance for partial synchrony settings.
These DAG-based consensus protocols~\cite{GagolLSS19,danezis2018blockmania,DAG-rider,Narwhal,spiegelman2022bullshark,fino}
form a backbone sequence of \leader blocks on the DAG.
However, as discussed previously, the current generation of DAG-riding solutions incurs increased latency and a complicated logic.
Very recently there has been an effort to reduce latency by interleaving two instances of Bullshark on the same DAG
~\cite{mystenPrivate,spiegelman2023shoal}, resulting in one \cbc latency reduction.
Cordial Miners is an alternative approach to reducing latency in DAG protocols by using only causal Best-Effort Broadcast (BEB).

\system borrows the key ideas that deal with partial synchrony from the above works,
modifying it to substantially reduce latency and simplify the logic.
Importantly, \system gets rid of specialized layers altogether, leading to an arguably simpler implementation.
It further reduces latency by integrating voting in the \bbca broadcast of a backbone block and by allowing non-backbone blocks to use BEB\@.

A key ingredient in DAG-based solutions is the notion of the secure, causal, and reliable broadcast, introduced by Reiter and Birman~\cite{reiter1994securely}, and later refined by Cachin~\cite{cachin2001secure} and Duan et al.~\cite{duan2017secure}.
This primitive was utilized in a variety of BFT replicated systems,
but not necessarily in the form of ``zero message overhead'' protocols riding on a DAG.
Several of these pre-blockchain-era BFT protocols are DAG based,
notably Total~\cite{moser1999byzantine} and ToTo~\cite{dolev1993early}, both of which are BFT solutions for the asynchronous model.

GradedDAG~\cite{gradedDAG} is a randomized consensus protocol that inspired us to abstract a core primitive whose strong
guarantees simplify the DAG-level consensus logic. GradedDAG uses graded Byzantine reliable broadcast (GRBC), an \rbc variant that yields graded output, to facilitate
consensus, whereas \system uses \bbca with \Complete-Adopt exposed from a \cbc primitive whose weaker properties suffice for consensus.
\system operates under the partial synchrony model whereas GradedDAG is asynchronous.

The intuition behind \Complete-Adopt is a notion at the heart of consensus solutions, familiar to ``Commit-Adopt'' by Gafni et al.\ in~\cite{gafni1998round}.
In the context of \bbca, it stipulates that a broadcast can \complete only after a quorum has locked it.
Although many forms of broadcasts have been formulated, including gradecast with multiple grades of output,
to our knowledge no previous broadcast primitive has been formulated with a \Complete-Adopt guarantee.

Finally, we remark that data availability is an orthogonal problem which already in
PBFT~\cite{PBFT} was addressed by batching transactions hashes into blocks and form consensus on blocks.
Recently, Narwhal~\cite{Narwhal} decoupled data availability from total ordering for DAG-based protocols.
In parallel, DispersedLeder~\cite{yang2022dispersedledger} and, later, HotShot~\cite{hotshot}
proposed a data-availability component which
pre-disseminates transactions and obtains a certificate of uniqueness and availability.

\begin{credits}
	\subsubsection*{\ackname} We would like to thank Nibesh Shrestha for his valuable assistance in refining the protocol presentation and pseudocode.
\end{credits}

\newpage

\bibliographystyle{abbrv}
\bibliography{bibliography}

\end{document}